\documentclass[runningheads]{llncs}
\usepackage{graphicx}
\usepackage{enumitem}
\usepackage{makeidx}
\usepackage{amsmath,amsfonts,amssymb}
\usepackage[T1]{fontenc} 
\usepackage[utf8]{inputenc}

\usepackage{cancel}
\usepackage{hyperref}
\usepackage{tikz, pgf}
\usetikzlibrary{arrows,automata,backgrounds}

\usepackage{float}
\usepackage{subfigure}

\definecolor{cgreen}{rgb}{0.0, .4, 0.0}

\newcommand{\pf}[1]{\text{Pref}(\mathcal #1)}

\newcommand{\ndashv}{\;\cancel\dashv\;}
\newcommand{\la}[1]{\mathcal L(\mathcal #1)}

\begin{document}
\title{Ordering regular languages: a danger zone}
%
%
\author{Giovanna D'Agostino \and
Davide Martincigh \and
Alberto Policriti}
\authorrunning{G. D'Agostino et al.}
%
\institute{University of Udine, Italy}
\maketitle              
\begin{abstract}
Ordering the collection of states of a given automaton starting from an order of the underlying alphabet is a natural move towards a computational treatment of the  language accepted by the automaton. Along this path, 
Wheeler \emph{graphs} have been recently introduced as an extension/adaptation of the Burrows-Wheeler Transform (the now famous BWT, originally defined on strings) to graphs. These graphs constitute an important data-structure for languages, since they allow a very efficient storage mechanism for the transition function of an automaton, while providing a fast support to  all sorts of  substring queries. This is possible as a consequence of a property---the so-called \emph{path coherence}---valid on Wheeler graphs and consisting in an ordering on nodes that ``propagates'' to (collections of) strings. By looking at a Wheeler graph as an automaton, the ordering on strings corresponds to the co-lexicographic order of the words entering each state. This leads naturally to consider the class of regular languages accepted by Wheeler automata, i.e. the Wheeler languages.

It has been shown that, as opposed to the general case, the classic determinization by powerset construction is polynomial on Wheeler languages. As a consequence, most of the classical problems  turn out to be ``easy''---that is, solvable in polynomial time---on Wheeler languages. Moreover,  deciding whether a DFA is Wheeler and deciding whether a DFA accepts a Wheeler language is polynomial. 

Our contribution here is to put an upper bound to easy problems. For instance, whenever we generalize by switching to general NFAs or by not fixing an order of the underlying alphabet, the above mentioned problems become ``hard''---that is NP-complete or even PSPACE-complete.

\end{abstract}
\section{Introduction}
\label{intro}

Adding an order to a class of structures (graphs, groups, monoids...) is a natural move. Moreover, ordering is a basic data-structuring mechanism, strongly favoring  computational manipulations of the considered class. 

Over the class of finite automata, order has been added e.g. in \cite{ordered_aut}, where the order must propagate along equally labeled transitions. 
Much more recently, in an effort to find a common denominator to a number of different algorithmic techniques, Gagie et al.  proposed (in \cite{Gagie}) a  new simple strategy for enforcing and using  order on a given automaton. Starting from an underlying order of the alphabet, the order on the states (formally given in Definition \ref{WheelerAutomaton}) must: i) agree with the order of the labels of their incoming edges, and ii) be coherent on target/source nodes, for pairs of arcs with equal labels. 
It turns out that this kind of automata, called Wheeler automata, (a) admit an efficient index data structure for searching  subpaths labeled with a given query pattern, and (b)  enable a representation of the graph in a space proportional to that of the edges' labels since the topology can be encoded with just $O(1)$ bits
per node \cite{Gagie} (as well as enabling more advanced compression mechanisms, see \cite{alanko2019tunneling,prezza2020locating}). 
This is in contrast with the fact that general graphs require a logarithmic (in the graph's size) number of bits per edge to be represented, as well as with recent results showing that in general, the subpath search problem can not be solved in subquadratic time, unless the strong exponential time hypothesis is false \cite{backurs2016regular,equi2020complexity,equi2020graphs,gibney2020simple,potechin2018lengths}.

\begin{figure}[ht]%
\begin{center}
\begin{tikzpicture}[->,>=stealth', semithick, initial text={}, auto, scale=.34]
 \node[state, label=above:{}, initial] (0) at (0,0) {$q_0$};
 \node[state, label=above:{}, accepting] (1) at (5,3) {$q_1$};
 \node[state, label=above:{}, accepting] (2) at (10,3) {$q_2$};
 \node[state, label=above:{}] (4) at (5,-3) {$q_4$};
 \node[state, label=above:{}] (3) at (10,-3) {$q_3$};
 \node[state, label=above:{}, accepting] (5) at (15,-3) {$q_5$};

\draw (0) edge [above] node [above] {$a$} (1);
\draw (1) edge [below] node [below] {$c$} (2);
\draw (2) edge [loop above] node [above] {$c$} (N1);
\draw (0) edge [bend left=0, above] node [bend right, above] {$d$} (4);
\draw (4) edge [above] node [above] {$c$} (3);
\draw (3) edge [loop below] node [above, xshift=8] {$c$} (3);
\draw (3) edge[above] node [above] {$f$} (5);
\draw (4) edge[bend left=40, above] node {$f$} (5);
\end{tikzpicture}
\end{center}
    \caption{A WDFAs $\mathcal A$ recognizing the language $\mathcal L_d = ac^*+dc^*f$. Condition (i) of Definition \ref{WheelerAutomaton} implies input consistency and induces the partial order $q_1 < q_2, q_3 < q_4 < q_5$. From condition (ii) it follows that $\delta(q_1, c) \le \delta(q_4, c)$, thus $q_2 < q_3$. Therefore, the only order that could make $\mathcal A$ Wheeler is $q_0 < q_1 < q_2 < q_3 < q_4 < q_5$. The reader can verify that condition (ii) holds for each pair of equally labeled edges.}%
    \label{w}%
\end{figure}
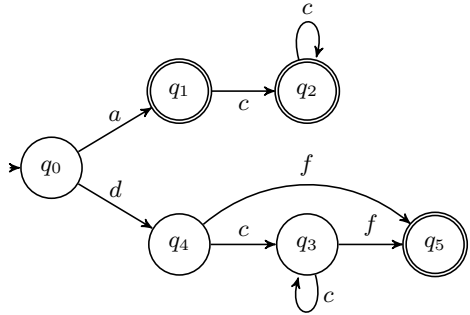

In Figure \ref{w} is depicted an example of a Wheeler automaton (see also Definition \ref{WheelerAutomaton}). Notice that the minimum DFA recognizing $\mathcal L$ is not input consistent, hence not Wheeler; we will return on this point later. 
One is naturally led to consider the class of Wheeler languages, i.e. the regular languages recognized by some Wheeler automaton, as well as to raise the question of whether it is possible to decide, just by looking at a (generic, possibly not Wheeler) DFA, whether the language it recognizes is Wheeler. Wheeler languages have been studied extensively in \cite{ADPP2}, where it is shown that we can decide in polynomial time whether a DFA recognizes a Wheeler language. 
Unfortunately, trying to use the same strategy on a NFA language does not work. In fact, in this paper we  show that  the problem for NFAs becomes PSPACE-complete.
  
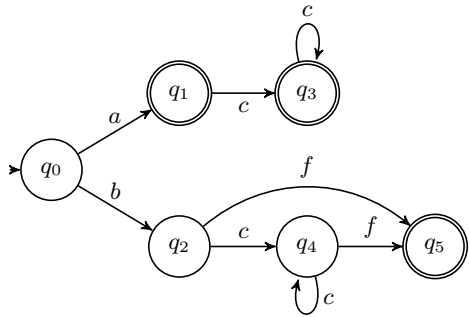
\begin{figure}[ht!]%
\begin{center}
\begin{tikzpicture}[->,>=stealth', semithick, initial text={}, auto, scale=.34]
 \node[state, label=above:{}, initial] (0) at (0,0) {$q_0$};
 \node[state, label=above:{}, accepting] (1) at (5,3) {$q_1$};
 \node[state, label=above:{}, accepting] (2) at (10,3) {$q_3$};
 \node[state, label=above:{}] (4) at (5,-3) {$q_2$};
 \node[state, label=above:{}] (3) at (10,-3) {$q_4$};
 \node[state, label=above:{}, accepting] (5) at (15,-3) {$q_5$};

\draw (0) edge [above] node [above] {$a$} (1);
\draw (1) edge [below] node [below] {$c$} (2);
\draw (2) edge [loop above] node [above] {$c$} (N1);
\draw (0) edge [bend left=0, above] node [bend right, above] {$b$} (4);
\draw (4) edge [above] node [above] {$c$} (3);
\draw (3) edge [loop below] node [above, xshift=8] {$c$} (3);
\draw (3) edge[above] node [above] {$f$} (5);
\draw (4) edge[bend left=40, above] node {$f$} (5);
\end{tikzpicture}
\end{center}
    \caption{A failed attempt to build a WDFA recognizing the language $\mathcal L_b = ac^*+bc^*f$. Condition (i) of Definition \ref{WheelerAutomaton} induces the partial order $q_1 < q_2 < q_3,q_4 < q_5$, and condition (ii) implies $\delta(q_1, c) \le \delta(q_2, c)$, thus $q_3 < q_4$. But if we apply condition (ii) again, we obtain $\delta(q_2, c) \le \delta(q_3, c)$, a contradiction.}%
    \label{nw}%
\end{figure}

Consider  the example in Figure \ref{nw} and confront it with Figure \ref{w}. 
Even tough apparently we made an insignificant change to the automaton, that is we have simply swapped character-labels $d$ and $b$, the new language recognized by the automaton turns out to be \emph{not} Wheeler. This example illustrates that the order of the underlying alphabet plays a significant role when determining whether an automaton, and even a language, is Wheeler. Notice that swapping $d$ and $b$ can be seen as considering a different order on $\Sigma$: from $a \prec c \prec d \prec f$ to $a \prec d \prec c \prec f$. Again a  natural question is raised: when presented with an automaton (or a language) that is \emph{not} Wheeler, is there a different way to order the alphabet so that the automaton (language) becomes Wheeler? Being able to answer in polynomial time to this this question would be very helpful for applications. Unfortunately, as shown in Section \ref{secGW}, this problem turns out to be NP-complete even for DFAs. 

WNFAs have the following useful property: turning a WNFA into a DFA using the (classic) powerset construction actually results in a WDFA with at most twice the number of states of the original WNFA. 
In other words, the  blow-up of states that we can observe  when converting NFAs into DFAs, does not occur for Wheeler non-deterministic automata.
Nevertheless, it is known (see \cite{ADPP}) that a blow-up of  states can occur even when we switch from the \emph{minimum} DFA recognizing a language $\mathcal L$ to the minimum WDFA recognizing $\mathcal L$. 
As a last contribution, in this paper we give an answer to a question put forward in \cite{ADPP}. We provide an algorithm to compute the minimum WDFA starting from the minimum DFA. Our algorithm  works in exponential time in the worst case and, since the dimension of the output might be exponential in the dimension of the input, is  optimal for the task. 

Due to space constraints, we present only the sketches of the proofs of our results. The full proofs, complete with all the missing details, can be found in the Appendix.

\section{Wheeler languages}
First of all, we fix some notation. Let $\Sigma$ denote a finite alphabet endowed with a total order $(\Sigma,\prec)$. We denote by $\Sigma^*$ the set of finite words over $\Sigma$, with $\varepsilon$ being the empty word, and we extend the order $\prec$ over $\Sigma$ to the \emph{co-lexicographic} order $(\Sigma^*, \prec)$, where $\alpha \prec \beta$ if and only the reverse of $\alpha$, i.e. $\alpha$ read from the right to the left, precedes lexicographically the reverse of $\beta$. Given two words $\alpha, \beta \in \Sigma^*$, we denote by $\alpha \dashv \beta$ the property that $\alpha$ is a suffix of $\beta$. 
For a language $\mathcal L \subseteq \Sigma^*$, we denote by $\pf L$ the set of prefixes of strings in $\mathcal L$.
We denote by $\mathcal A = (Q, q_0, \delta, F, \Sigma)$ a finite automaton (NFA), with $Q$ as set of states, $q_0$ initial state, $\delta: Q \times \Sigma \rightarrow 2^Q$ transition function, and $F \subseteq Q$ final states. 
The dimension of $\mathcal A$, denoted by $|\mathcal A|$, is defined as the number of states of $\mathcal{A}$.
An automaton is deterministic (DFA) if $|\delta(q, a)| \le 1$, for all $q\in Q$ and $a\in \Sigma$. As customary, we extend $\delta$ to operate on strings as follows: for all $q\in Q$, $a\in \Sigma$ and $\alpha \in \Sigma^*$
\[
\delta(q,\varepsilon) = \{q\}, \qquad \delta(q,\alpha a)=\bigcup_{v\in \delta(q,\alpha)} \delta(v,a). 
\]
We denote by $\la A = \{\alpha \in \Sigma^*:\, \delta(q_0,\alpha) \cap F \ne \emptyset\}$ the language accepted by the automaton $\mathcal A$.
We make the assumption that every automaton is basic, that is, every state is reachable from the initial state and every state can reach at least one final state. Notice that this assumption is not restrictive, since removing from a NFA every state not reachable from $q_0$ and every state from which is impossible to reach a final state can be done in linear time and does not change the accepted language.
It immediately follows that: 
\begin{itemize}
    \item there might be only one state without incoming edges, namely $q_0$; 
    \item every word that can be read starting from $q_0$ belongs to $\pf L$.
\end{itemize}
Lastly, given a finite set $Q$, totally ordered by the relation $<$, we say that $I \subseteq Q$ is an interval if and only if for all $q,q',q'' \in Q$ with $q\le q'\le q''$, if $q,q''\in I$ then $q' \in I$. We denote by $I:=[q_{\min}, q_{\max}]$ a generic interval, where $q_{\min}$ ($q_{\max}$) is the minimum (maximum) element of $I$ with respect to <.

The class of Wheeler automata has been recently introduced in \cite{Gagie}. An automaton in this class has the property that there exists a total order on its states that is propagated along equally labeled transition. Moreover, the order must be compatible with the underlying order of the alphabet:

\begin{definition}[Wheeler Automaton]
\label{WheelerAutomaton}
A Wheeler NFA (WNFA) $\mathcal{A}$ is a NFA $(Q,q_0,\delta,F,\Sigma)$  endowed with a binary relation
<, such that: $(Q,<)$ is a linear order having the initial state $q_0$ as minimum, $q_0$ has no in-going edges, and
the following two (Wheeler) properties are satisfied. Let $v_1 \in \delta(u_1, a_1)$ and $v_2 \in \delta(u_2, a_2)$:
\begin{enumerate}[label = (\roman*)]
    \item $a_1 \prec a_2 \,\rightarrow \, v_1 < v_2$ 
    \item $(a_1 = a_2 \wedge u_1 < u_2) \,\rightarrow \, v_1 \le v_2$.
\end{enumerate}
A Wheeler DFA (WDFA) is a WNFA in which the cardinality of $\delta(u,a)$ is always less than or equal to
one.
\end{definition}

In Figure \ref{w} is depicted an example of a WDFA.

\begin{remark}
A consequence of Wheeler property (i) is that A is \emph{input-consistent}, that is all transitions
entering a given state $u \in Q$ have the same label: if $u \in \delta(v,a)$ and $u \in \delta(w,b)$, then $a=b$. Therefore the function $\lambda: Q\rightarrow \Sigma$ that associate to each state the unique label of its incoming edges is well defined. For the state $q_0$, the only one without incoming edges, we set $\lambda(q_0) := \#$. 
\end{remark}

In \cite{Gagie} it is shown that WDFAs have a property called \emph{path coherence}: let $\mathcal A = (Q,q_0,\delta,F,\Sigma)$ be a WDFA according to the order $(Q,<)$. 
Then for every interval of states $I=[q_i, q_j]$ and for all $\alpha \in \Sigma^*$, the set $J$ of states reachable starting from any state of $I$ by reading $\alpha$ is also an interval.
\emph{Path coherence} allows us to transfer the order < over the states of $Q$ to the co-lexicographic order $\prec$ over the words entering the states: for each state $q$ define the set $I_q = \{\alpha:\,\delta(q_0,\alpha)=q)\}$. 
With abuse of notation, we extend $\preceq$ to subsets of $\Sigma^*$ as follows: $U \preceq V$, with $U,V \subseteq \Sigma^*$, if and only if $\alpha\prec \beta$ for each $\alpha\in U,\ \beta\in V$ such that $\{\alpha,\beta\} \notin U \cap V$.
Then, two states $q$ and $p$ with $I_q \ne I_p$ satisfy $q < p$ if and only $I_q \preceq I_p$ (again proved in \cite{ADPP}).
An immediate consequence of this fact is that a WDFA admits an unique order of its states that makes it Wheeler and this order is univocally determined by the co-lexicographic order of any word entering its states.
This result is important for two different reasons. First of all, it makes possible to decide in polynomial time whether a DFA is Wheeler: for each state $q$, pick a word $\alpha_q$ entering it and order the states reflecting the co-lexicographic order of the words \{$\alpha_q:\, q \in Q\}$; then check if the order satisfies the Wheeler conditions. 
Secondly, it is the key to adapt Myhill-Nerode Theorem to Wheeler automata. We recall the following defintion.

\begin{definition}[Myhill-Nerode equivalence]
\label{equivL}
Let $\mathcal L \subseteq \Sigma^*$ be a language. Given a word $\alpha \in \Sigma^*$, we define the \emph{right context} of $\alpha$ as  
\[
\alpha^{-1}\mathcal L := \{\gamma \in \Sigma^*:\, \alpha\gamma \in \mathcal L\},
\]
and we denote by $\equiv_\mathcal L$ the Myhill-Nerode equivalence on $\pf L$ defined as
\[
\alpha \equiv_\mathcal L \beta \iff \alpha^{-1}\mathcal L = \beta^{-1}\mathcal L.
\]
\end{definition}

The (classic) Myhill-Nerode Theorem, among many other things, establishes a bijection between equivalence classes of $\equiv_\mathcal L$ and the states of the minimum DFA recognizing $\mathcal L$. This minimum automaton is also unique up to isomorphism and a similar result, fully proved in \cite{ADPP2}, holds for Wheeler languages as well. 

In order to state such an analogous of Myhill-Nerode Theorem for  Wheeler languages, the equivalence $\equiv_\mathcal L$ is replaced by the equivalence $\equiv_\mathcal L^c$ defined below.

\begin{definition}
The input consistent, convex refinement $\equiv_\mathcal L^c$ of $\equiv_\mathcal L$ is defined as follows. $\alpha \equiv_\mathcal L^c \beta$ if and only if
\begin{itemize}
    \item $\alpha \equiv_\mathcal L \beta$,
    \item $\alpha$ and $\beta$ end with the same character, and 
2

    \item for all $\gamma \in \pf L$, if $\min(\alpha, \beta) \preceq \gamma \preceq \max(\alpha,\beta)$, then $\alpha \equiv_\mathcal L\gamma \equiv_\mathcal L \beta$.
\end{itemize}
\end{definition}

The  Myhill-Nerode Theorem for Wheeler languages proves that there exists a minimum (in the number of states) WDFA recognizing $\mathcal L$. As in the classic case,  states of the minimum automaton are, in fact, $\equiv_\mathcal L^c$-equivalence classes, this time consisting of \emph{intervals} of words. Also such WDFA is unique up to isomorphism.   

A further important consequence, especially for testing Wheelerness, is stated in the following Lemma (again proved in \cite{ADPP2}).

\begin{lemma}
\label{monotone}
A regular language $\mathcal L$ is Wheeler if and only if all monotone sequences in $(\pf L, \prec)$ become eventually constant modulo $\equiv_\mathcal L$. In other words, for all sequences $(\alpha_i)_{i \ge 0}$ in $\pf L$ with
\[
\alpha_1 \preceq \alpha_2 \preceq \dots \alpha_i \preceq \dots \quad \text{ or }\quad \alpha_1 \succeq \alpha_2 \succeq \dots \succeq \alpha_i \succeq \dots
\]
there exists an $n$ such that $\alpha_h \equiv_\mathcal L \alpha_k$, for all $h,k \ge n$. 
\end{lemma}

Lemma \ref{monotone} shows how it is possible to recognize whether a language $\mathcal L$ is Wheeler simply by verifying a property on the words of $\pf L$: trying to find a WDFA that recognizes $\mathcal L$ is no longer needed to decide the Whelerness of $\mathcal L$. 
As it turns out, we can verify whether the property depicted in Lemma \ref{monotone} is satisfied just by looking at \emph{any} DFA recognizing $\mathcal L$, as shown in Theorem \ref{polynomialW} (see \cite{ADPP2}).

\begin{theorem}
\label{polynomialW}
Let $\mathcal A$ be a DFA such that $\mathcal L = \la A$,  with initial state $q_0$ and dimension $n = |\mathcal A|$. 
\\$\mathcal L$ is not Wheeler if and only if there exist $\mu, \nu$ and $\gamma$ in $\Sigma^*$, with $\gamma \ndashv \mu,\nu$, such that:
\begin{enumerate}
    \item $\mu \not\equiv_\mathcal L \nu$ and they label paths from $q_0$ to states $u$ and $v$, respectively;
    \item $\gamma$ labels two cycles, one starting from $u$ and one starting from $v$;
    \item $\mu, \nu \prec \gamma$\; or \; $\gamma \prec \mu,\nu$.
\end{enumerate}
The length of the words $\mu, \nu$ and $\gamma$ satisfying the above  can be bounded: 
\begin{enumerate}
\setcounter{enumi}{3}
    \item $|\mu|, |\nu| \le |\gamma| \le n^3+2n^2+n+2$.
\end{enumerate}
\end{theorem}

Since in this work we make an extensive use of Theorem \ref{polynomialW}, here is a simple example on how and why it works. Consider the automata depicted in Figure \ref{exth}.  
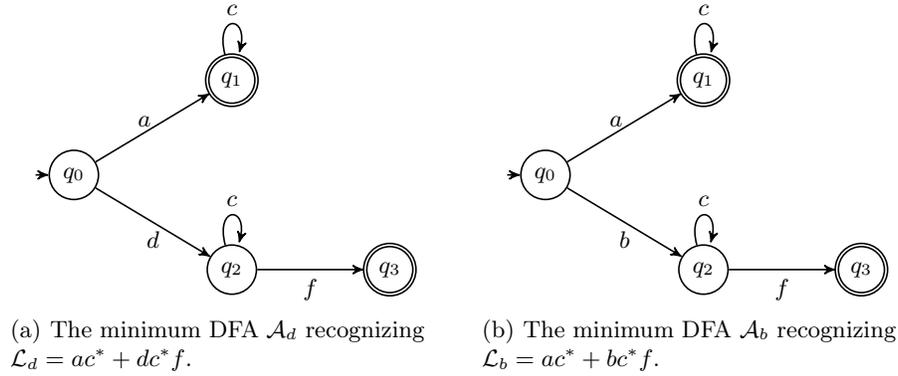
\begin{figure}[ht]%
    \centering
    \subfigure[The minimum DFA $\mathcal A_d$ recognizing $\mathcal L_d = ac^*+dc^*f$.] {
    \label{mwa}
\begin{tikzpicture}[->,>=stealth', semithick, initial text={}, auto, scale=.42]
 \node[state, minimum size=2pt, initial] (0) at (0,0) {$q_0$};
 \node[state, label=above:{}, minimum size=2pt, accepting] (1) at (5,3) {$q_1$};
 \node[state, label=above:{}, minimum size=2pt] (2) at (5,-3) {$q_2$};
 \node[state, label=above:{}, minimum size=2pt, accepting] (3) at (10,-3) {$q_3$};

\draw (0) edge [above] node [above, xshift=-3, yshift=-3] {$a$} (1);
\draw (0) edge [below] node [below] {$d$} (2);
\draw (1) edge [loop above] node {$c$} (1);
\draw (2) edge [loop above] node {$c$} (2);
\draw (2) edge [above] node [below] {$f$} (3);
\end{tikzpicture}
    }
    \qquad
\subfigure[The minimum DFA $\mathcal A_b$ recognizing $\mathcal L_b = ac^*+bc^*f$.] {
    \label{mnwa}
\begin{tikzpicture}[->,>=stealth', semithick, initial text={}, auto, scale=.42]
 \node[state, minimum size=2pt, initial] (0) at (0,0) {$q_0$};
 \node[state, label=above:{}, minimum size=2pt, accepting] (1) at (5,3) {$q_1$};
 \node[state, label=above:{}, minimum size=2pt] (2) at (5,-3) {$q_2$};
 \node[state, label=above:{}, minimum size=2pt, accepting] (3) at (10,-3) {$q_3$};

\draw (0) edge [above] node [above, xshift=-3, yshift=-3] {$a$} (1);
\draw (0) edge [below] node [below] {$b$} (2);
\draw (1) edge [loop above] node {$c$} (1);
\draw (2) edge [loop above] node {$c$} (2);
\draw (2) edge [above] node [below] {$f$} (3);
\end{tikzpicture}
}%
    \caption{The minimum DFAs recognizing the languages $\mathcal L_d$ (Wheeler) and $\mathcal L_b$ (not Wheeler).}%
    \label{exth}%
\end{figure}
As shown in Figure \ref{w} and \ref{nw}, the language $\mathcal L_d$ is Wheeler whereas $\mathcal L_b$ is not, as can be easily proved using Theorem \ref{polynomialW}. In fact, consider the automaton $\mathcal A_b$. 
By setting $\mu := a$, $\nu := b$ and $\gamma := c$, one can  verify  conditions 1-3 of the theorem are satisfied. Notice that  condition $a \not\equiv_{\mathcal L_b} b$ follows immediately from the fact that $\delta(q_0, a) \ne \delta(q_0, b)$ in the minimum DFA $\mathcal A_b$. 

If we try to transpose the same reasoning to the automaton $\mathcal A_d$ by setting $\mu = a$, $\nu = d$ and $\gamma = c$, condition 3 of Theorem \ref{polynomialW} is no longer satisfied. We can not find 3 words satisfying conditions 1-3 of Theorem \ref{polynomialW}, therefore $\mathcal L_b$ is Wheeler.

\vskip 5mm

The polynomial bound given by condition 4 of Theorem \ref{polynomialW} allows us to design an algorithm that decides whether a given DFA recognizes a Wheeler language: using dynamic programming (see \cite{ADPP}), it is possible to keep track of all the relevant paths and cycles inside the DFA and check, in polynomial time, whether there exists three words satisfying the conditions of the theorem.  

Things change if, instead of a DFA, we are given a NFA. Trying to exploit the same idea used for DFAs does not work: the problem of deciding whether two words $\mu$ and $\nu$ read by a NFA are Myhill-Nerode equivalent is PSAPCE-complete, whereas is polynomial on DFAs (simply compute the minimum DFA using Hopcroft's algorithm and check whether the two states reached by $\mu$ and $\nu$ coincide). Even worse, the straightforward attempt of building the minimum DFA recognizing the NFA's language might lead to a blow-up of the sates, resulting in a exponential time (and exponential space) algorithm. 

We show that the problem of deciding whether a NFA recognizes a Wheeler language is indeed hard, but doesn't require exponential time to be solved: the problem turns out to be PSPACE-complete. To show this, we first need to adapt Theorem \ref{polynomialW} to work on NFAs, as described in the following corollary.

\begin{corollary}
\label{nfa length}
Let $\mathcal A = (Q, q_0, \delta, F, \Sigma)$ be a NFA of dimension $n := |\mathcal A|$. Then $\mathcal L := \la A$ is not Wheeler if and only if there exist three words $\mu, \nu, \gamma$, with $\gamma \ndashv \mu,\nu$, such that
\begin{enumerate}
    \item $\mu\gamma^i \not\equiv_{\mathcal L} \nu\gamma^j$ for all $0 \le i,j \le 2^n$; 
    \item $\gamma$ labels two cycles, one starting from a state $p \in \delta(q_0,\mu)$ and one from a state $r \in \delta(q_0,\nu)$;
    \item $\mu, \nu \prec \gamma$ or $\gamma \prec \mu, \nu$.
\end{enumerate}
The length of the words $\mu, \nu$ and $\gamma$ satisfying the above can be bounded: 
\begin{enumerate}
\setcounter{enumi}{3}
    \item $|\mu|, |\nu| < |\gamma| \in O(n^3\cdot2^{3n})$.
\end{enumerate}
\end{corollary}
\begin{proof}[Sketch]
 $(\Longrightarrow)$ Let $\mathcal D$ be the minimum DFA recognizing the language $\mathcal L$, with initial state $\hat q_0$. Recall that $\mathcal D$ might have up to $2^n$ states. From Theorem \ref{polynomialW} we know that $\mathcal L$ is not Wheeler if and only if there exist $\hat\mu, \hat\nu$ and $\hat\gamma$ satisfying conditions 1-4 of Theorem \ref{polynomialW}. 
 Although there is not a perfect correspondence between cycles in $\mathcal A$ and cycles in $\mathcal D$, whenever we find in $\mathcal D$ a path $\hat\mu$ (respectively, $\hat\nu$) and a cycle $\hat\gamma$ both ending in the same state $\hat u$ ($\hat v$), we can find in $\mathcal A$ a path $\mu = \hat\mu \hat\gamma^i$  ($\nu = \hat\nu \hat\gamma^j$) and a cycle $\gamma_1 = \hat\gamma^{k_1}$ ($\gamma_2 = \hat\gamma^{k_2})$ both ending in the same state $u$ ($v$), for some $i,j,k_1,k_2 \le n$. 
 Note that the word $\hat\gamma^{(n+1)k_1k_2}$ labels two cycles starting from $u$ and $v$. Moreover, we have 
 \[
 |\hat\gamma^{(n+1)k_1k_2}| \ge (n+1) |\hat\gamma^{k_1k_2}| \ge (n+1) |\hat\gamma| \ge (i+1) |\hat\gamma| > |\hat\mu\hat\gamma^i|,
 \]
 and similarly $|\hat\gamma^{(n+1)k_1k_2}|> |\hat\nu\hat\gamma^j|$. Since $|\hat\gamma| \in O(2^{3n})$, we have $|\hat\gamma^{(n+1)k_1k_2}| \in O(n^3 \cdot 2^{3n})$ and therefore the words $\mu, \nu$ and $\gamma := \hat\gamma^{(n+1)k_1k_2}$ satisfy conditions 2-4 of this corollary. Finally, condition 1 is satisfied for all $i,j \ge 0$ because when reading $\mu\gamma^i$ and $\nu\gamma^j$ on $\mathcal D$, they reach different states.\\ 
 $(\Longleftarrow)$ 
 As we did while proving the opposite direction, given three words $\mu, \nu, \gamma$ in $\mathcal A$ satisfying condition 1-3, we can always find, for some $i,j,k \le 2^n$, three words $\hat\mu = \mu\gamma^i$, $\hat\nu = \nu\gamma^j$ and $\hat\gamma = \gamma^k$ in $\mathcal D$ that almost satisfy condition 1-3 of Theorem \ref{polynomialW}. The only requirement that might be missing is that $\hat\mu \not\equiv_\mathcal L \hat\nu$. 
 The upper bound $2^n$ in condition 1 ensures that this requirement is fulfilled, hence we can apply Theorem \ref{polynomialW} to conclude that $\mathcal L$ is not Wheeler.
\end{proof}

Despite the fact the the bound in condition 4 has become exponential by switching to NFAs, it is still possible to check in polynomial space (but exponential time) whether there are three words $\mu, \nu$ and $\gamma$ satisfying the conditions of Corollary \ref{nfa length}. Thus we can prove the following:

\begin{proposition}
\label{pspace}
Given a NFA $\mathcal A$, deciding whether the language $\mathcal L := \la A$ is Wheeler is PSPACE-complete.
\end{proposition}
\begin{proof}[Sketch]
To prove that the problem is in PSPACE (=NPSPACE), we guess (using non-determinism) three words $\mu, \nu$ and $\gamma$ satisfying the conditions of Corollary \ref{nfa length}. Since this words are too long to be stored, we first guess their length bit by bit and then we guess their characters one by one, starting from the last one and proceeding from right to left. Meanwhile, for all $\alpha \in \{\mu, \nu, \gamma\}$ and for all state $t$ of $\mathcal A$, we follow backwards the edges of $\mathcal A$ labeled as the character of $\alpha$ that we are guessing. This way, we can compute the set of states from which is possible to reach $t$ by reading $\alpha$. This allows us to find, for all $i,j \le 2^n$, the sets $Q^i_{\mu} := \delta(q_0, \mu\gamma^i)$ and $Q^j_\nu := \delta(q_0, \nu\gamma^j)$. We can then test whether $\mu\gamma^i \not\equiv_\mathcal L \nu\gamma^j$ by confronting the language accepted by the automaton $\mathcal A$ with set of initial states $Q^i_{\mu}$ and the language accepted by the automaton $\mathcal A$ with set of initial states $Q^i_{\nu}$. This can be done in polynomial space since the problem of deciding whether two NFAs recognize the same language is known to be PSPACE-complete.
\vskip 5mm
To prove the hardness of the problem, we show a polynomial reduction from the universality problem for NFA, i.e. the problem of deciding whether the language accepted by a NFA $\mathcal A$ over the alphabet $\Sigma$ is such that $\la A = \Sigma^*$.  
\\Let $\mathcal A = (Q, q_0,\delta, F, \Sigma)$ be a NFA and let $\mathcal L := \la A$. We can assume, without loss of generality, that $q_0 \in F$, otherwise $\mathcal A$ would not accept the empty word and we could immediately derive that $\mathcal L \ne \Sigma^*$. Starting from $\mathcal A$, we build in constant time a NFA $\mathcal A''$ over the alphabet $\Sigma \cup \{a,b,c\}$ that recognizes the language 
\[\mathcal L'' := a\cdot(\mathcal Lc)^*\cdot\mathcal L + b\cdot (\Sigma + c)^*,
\]
where $a,b,c$ are three characters not in $\Sigma$ and such that $a \prec b \prec c$ (the order of the characters of $\Sigma$ is irrelevant in this proof). If we prove that $\mathcal L = \Sigma^*$ if and only if $\mathcal L''$ is Wheeler, the reduction is complete and the thesis follows. Notice that, since $\varepsilon \in \mathcal L$, the following property holds:
\begin{equation}
\label{eq:1}
\mathcal L = \Sigma^* \iff (\mathcal Lc)^*\cdot\mathcal L = (\Sigma+c)^*.
\end{equation}
Let us prove that $\mathcal L = \Sigma^*$ implies that $\mathcal L''$ is not Wheeler. If $\mathcal L = \Sigma^*$, then by \eqref{eq:1} we have $\mathcal L'' = (a+b)(\Sigma+c)^*$. The minimum DFA recognizing $\mathcal L''$ has only one cycle, therefore from Theorem \ref{polynomialW} it follows that $\mathcal L''$ is Wheeler.\\
We next prove that $\mathcal L''$ not Wheeler implies $\mathcal L=\Sigma^*$. Notice that $a^{-1}\mathcal L''=(\mathcal Lc)^*\cdot\mathcal L $ so that,  by \eqref{eq:1}
if  $\mathcal L \ne \Sigma^*$, then $a^{-1}\mathcal L''=(\mathcal Lc)^*\cdot\mathcal L \ne (\Sigma+c)^*$. However we have $b^{-1}\mathcal L'' = (\Sigma+c)^*$, thus $a \not\equiv_\mathcal{L''} b$. Moreover, it can be proved that, for all $n$, both $a \equiv_\mathcal{L''} ac^n$ and $b \equiv_\mathcal{L''} bc^n$ hold. Therefore the monotone sequence 
\[
ac \prec bc \prec ac^2 \prec bc^2 \prec \dots \prec ac^n \prec bc^n \prec \dots
\]
is not constant modulo $\equiv_\mathcal{L''}$ and from Lemma \ref{monotone} it follows that $\la{A''}$ is not Wheeler.
\end{proof}

\section{Generalized Wheelerness}
\label{secGW}

As we have already pointed out in the introduction, changing the underlying order of the alphabet might turn a Wheeler language into a not Wheeler one and vice versa. For instance, consider again the Wheeler languages $\mathcal L_d$ and the regular (but not Wheeler) language $\mathcal L_b$ depicted in Figure \ref{exth}. 
If we change the order of $\Sigma$ from $a \prec c \prec d \prec f$ to $a \prec d \prec c \prec f$, the Wheeler language $\mathcal L_d$ turns into a non-Wheeler language (isomorphic to $\mathcal L_b$ under the isomorphism $\varphi$ between alphabets that fixes characters $a,c,f$ and sends $d$ into $\varphi(d) = b$). Hence, by not fixing  an apriori order of the alphabet $\Sigma$ we enlarge the class of languages.

\begin{definition}[Generalized Wheelerness]
\label{defGW}
A NFA $\mathcal A$ over the alphabet $\Sigma$ is called a Generalized Wheeler Automaton (GWNFA) if and only if there exists an ordering of the elements of $\Sigma$ that makes $\mathcal A$ Wheeler.
\\A language $\mathcal L$ is called \emph{generalized Wheeler} (for short GW) if and only if there exists a GWNFA that recognizes $\mathcal L$.
\end{definition}

Let $\mathcal A$ be a WDFA. Then, every word $\alpha$ that labels a cycle in $\mathcal A$ is \emph{primitive} (see \cite{ADPP2}), that is there exists no $\beta \ne \varepsilon$ and $i>1$ such that $\alpha = \beta^i$. A direct consequence of this property is that Wheeler languages form a subclass of star-free languages, i.e. the class of languages that can be defined by a regular expression not containing the Kleene star. Since star-free expressions, and thus star-free languages, are closed under permutations of the alphabet, even GW languages must be a subclass of star-free languages. Here we show that the inclusion is strict, therefore GW languages must be studied  separately.

\begin{proposition}
If $|\Sigma| \ge 2$, then the set $\{ \mathcal L \subseteq \Sigma^*: \; \mathcal L \text{ is GW} \}$ is a proper subset of $\{ \mathcal L \subseteq \Sigma^*: \; \mathcal L \text{ is star-free} \}$.
\end{proposition}
\begin{proof}
Let $a,b$ be two distinct characters of $\Sigma$, and consider the language $\mathcal L=a(aba)^*a+ba(aba)^*b$. It is possible to prove that $\mathcal L$ is star-free, but $\mathcal L$ is not GW: consider the sequence $(\alpha_i)_{i\ge 2}$ with $\alpha_{2n}=a(aba)^n$ and $\alpha_{2n+1}=ba(aba)^n$. Since, for all $i$, the word $\alpha_i$ is a prefix of $\alpha_{i+1}$, independently from how $a$ and $b$ are ordered we have 
\[
aaba \prec baaba \prec aabaaba \prec baabaaba \prec \dots \prec a(aba)^i \prec ba(aba)^i \prec \dots
\]
Moreover, for all $i$ we have $a(aba)^i \not\equiv_\mathcal L ba(aba)^i$, since $a(aba)^i \cdot a$ belongs to $\mathcal L$ but $ba(aba)^i \cdot a$ does not. We can then apply Lemma \ref{monotone} to conclude that $\mathcal L$ is not Wheeler. Since this result does not depend on the order of the alphabet, $\mathcal L$ is not GW.
\end{proof}

We mentioned that we can decide in polynomial time whether a DFA is Wheeler. On the contrary, deciding whether a NFA is Wheeler is NP-complete, even when we bound the outdegree of each state of the NFA to be at most 5 (see \cite{NP}). 
As one may expect, deciding whether a NFA is a GWNFA is not easier than deciding whether a NFA is Wheeler. In fact, we show in Proposition \ref{W to GW} that the problem is NP-complete. 
We actually prove a stronger result in Proposition \ref{GWDFA}: even deciding whether a DFA is a GWNFA is NP-complete. 
Since the proof of Proposition \ref{GWDFA} is cumbersome, we decided to present also it weaker version, i.e. Proposition \ref{W to GW}, which shows a more natural reduction from the problem of deciding whether a NFA is Wheeler. Both proofs can be found in the Appendix.
It is worth noticing that the proof of Proposition \ref{W to GW} can be adapted to work even on DFAs, hence giving an alternative way to prove that deciding whether a DFA is a GWNFA is NP-complete. Nonetheless, Proposition \ref{GWDFA} is still stronger, since it also proves that deciding whether a DFA recognizes a GW language is NP-complete.

\begin{proposition}[GWNFA hardness]
\label{W to GW}
Let $\mathcal A$ be a NFA. Deciding whether $\mathcal A$ is a GWNFA is NP-complete.
\end{proposition}

\begin{proposition}[GWDFA and GW languages hardness]
\label{GWDFA}
Let $\mathcal L\subseteq \Sigma^*$ be a language and $\mathcal A$ be a DFA. Both the problems of deciding whether $\mathcal A$ is a GWNFA and deciding whether $\mathcal L$ is GW are NP-complete.
\end{proposition}

\section{DFA to WDFA}

As discussed in Section \ref{intro}, a Wheeler automaton can be represented more compactly than a generic finite automaton. Therefore, if we are given a Wheeler language $\mathcal L$ represented as a DFA $\mathcal A$ that recognizes $\mathcal L$, we may be tempted to look for a WDFA $\mathcal A'$ that recognizes the same language to achieve a better, i.e. more compact, representation. 
Unfortunately, this approach might not work in our favor: in \cite{ADPP}, a family $(\mathcal L_m)_{m\ge1}$ of Wheeler languages with the property that the dimension of the minimum WDFA recognizing $\mathcal L_m$ is exponential in the dimension of the minimum DFA recognizing $\mathcal L_m$ is presented. 
Note that we can always assume that, whenever we are given a DFA or a WDFA, the automaton is minimum. In fact, both tasks of minimizing a DFA and minimizing a WDFA can be done in polynomial time (see \cite{Hop} for DFAs and \cite{ADPP} for WDFAs).
Here we answer (positively) to the open question put forward in \cite{ADPP} whether there exists an algorithm to compute the minimum WDFA starting from the minimum DFA. Our algorithm works in exponential time in the worst case and, since the dimension of the output might be exponential in the dimension of the input, is optimal for the task. 

Let $\mathcal L = \la A$ be the language recognized by the given (minimum) DFA $\mathcal A$. Recall that there is a 1 to 1 correspondence between the $\equiv_\mathcal L^c$-classes and the states of the minimum WDFA recognizing $\mathcal L$. 
First of all, our algorithm identifies a representative for each $\equiv_\mathcal L^c$-class; to be able do this in exponential time, we first need to put a bound on the length of such representatives.

\begin{lemma}
\label{short}
Let $\mathcal A$ be the minimum DFA recognizing the Wheeler language $\mathcal L = \la A$ over the alphabet $\Sigma$, and let $C_1,...,C_m$ be the pairwise distinct equivalence classes of $\equiv_\mathcal L^c$. Then, for each $1 \le i \le m$, there exists a word $\alpha_i \in C_i$ such that $|\alpha_i| < n + n^2$, where $n := |\mathcal A|$.
\end{lemma}
\begin{proof}[Sketch] Suppose by contradiction that there exists a class $C_i$ such that for all $\alpha \in C_i$ it holds $|\alpha| \ge n + n^2$, and let $\alpha \in C_i$ be a word of minimum length. Since $\mathcal A$ has $n$ states, there must exists a factor of $\alpha$ that corresponds to a cycle in the run of $\alpha$ on $\mathcal A$. Let $\beta$ be the word obtained by erasing such factor from $\alpha$. Since $\alpha$ and $\beta$ ends in the same state of $\mathcal A$, we have $\alpha \equiv_\mathcal L \beta$. From the minimality of $\alpha$ it follows that $\alpha \not\equiv^c_\mathcal L \beta$, therefore there must exists a word $\eta$ that is not Myhill-Nerode equivalent to $\alpha$ and that is included (co-lexicographically) between $\alpha$ and $\beta$. We can further assume that the length of $\eta$ is greater than $n^2$, hence we can identify a common factor of $\alpha$ and $\eta$ that labels two cycles in the runs of $\alpha$ and $\eta$ on $\mathcal A$. We can then apply Theorem \ref{polynomialW} to conclude that $\mathcal L$ is not Wheeler, a contradiction.
\end{proof}

The algorithm works as follow: first of all, it generates the list containing every word of length less than $n+n^2$ that can be read on $\mathcal A$. This can be done in exponential time since the list contains at most $|\Sigma|^{n+n^2}$ words. 
Then, we order the list co-lexicographically and we apply the following proposition.

\begin{proposition}[Minimum DFA to minimum WDFA]
\label{DFA to WDFA}
Let $\mathcal A$ be the minimum automaton recognizing $\mathcal L = \la A$ with $| \mathcal A | = n$, over the alphabet $\Sigma$ with $|\Sigma|=\sigma$. Let $d:= n + n^2$ and define 
\[ \pf L^{\le d} := \{\alpha \in \pf L: \; |\alpha| \le d \}. \] 
Assume that we are given the elements of $\pf L^{\le d}$ in co-lexicographic order, i.e. $\pf L^{\le d} = \{ \alpha_1, ..., \alpha_k \}$ with $\alpha_i \prec \alpha_{i+1}$. Then it is possible to build the minimun WDFA recognizing $\mathcal L$ in $O(k+n^2\cdot\sigma\cdot m\log m)$ time, where $m$ is the dimension of the output.
\end{proposition}
\begin{proof}[Sketch]
From Lemma \ref{short} we know that each $\equiv_\mathcal L^c$-class has at least one representative in $\pf L^{\le d}$. By considering the list $[\alpha_1]_\mathcal L, \dots, [\alpha_k]_\mathcal L$ we can spot such representatives: $\alpha_i$ represents a $\equiv_\mathcal L^c$-class if and only if $\alpha_i$ and $\alpha_{i+1}$ belong to different $\equiv_\mathcal L$-classes or they belong to the same class but their last character is different. 

Each representative of a $\equiv_\mathcal L^c$-class become a state of our WDFA. To compute the edges of the automaton, we simply check, for all representative $\beta$ and for all $c \in \Sigma$, the $\equiv_\mathcal L^c$-class of $\beta \cdot c$.
\end{proof}

\section{Conclusions}

Having an order on the set of states of an automaton that is consistent with a given order of the underlying alphabet  has important implication from a practical point of view. In fact, deterministic or even non-deterministic finite state automata enjoying this property---Wheeler automata---can be used to efficiently analyse the language they accept by standard tools. 

In this work we established a few limitations along the directions that one can imagine to take in order to extend the nice properties of Wheeler languages.

More specifically, we proved that adding as a degree of freedom the possibility to re-order the underlying alphabet produces a significantly more complex (NP-complete) class of languages.  In addition, the polynomial test we have for testing whether a language is Wheeler when the latter is presented by a \emph{deterministic} automaton, turns out much more complex (PSPACE-complete) if the language is presented by a non-deterministic one. The picture is completed by explicitly giving an algorithm that turns a DFA into a Wheeler automaton, whenever this is possible.

Proving the above limitations should clarify the role played by apparently secondary aspects of the definition of Wheeler language. This kind of study should help to better understand the nature of Wheeler languages, with the ultimate goal of singling out any feature that may admit some sort of extension. 

Consider, in more general terms, the following open problem: is there a class of   automata $\mathcal A$  properly extending the class of Wheeler automata and such that their  deterministic equivalent $\mathcal D_{\mathcal A}$ is such that  $|\mathcal D_{\mathcal A}| \in \text{poly}(|\mathcal A|)$?

\newpage
\section{Appendix}

\renewenvironment{proof}[1][\proofname]{{\bfseries #1.}}{\qed}

\begin{proof}[Proof of Corollary \ref{nfa length}]
Let $\mathcal D = (\hat Q, \hat q_0, \hat \delta, \hat F, \Sigma)$ be the minimum DFA recognizing $\mathcal L$. Clearly $\mathcal D$ has at most $2^n$ states.
\\($\Longleftarrow$) From condition 2 it follows that $\mu\gamma^* \subseteq \pf L$, so consider the following list of $2^n+1$ states of $\mathcal D$:
\[
\hat\delta(\hat q_0, \mu\gamma^0), \, \hat\delta(\hat q_0, \mu\gamma^1), \, \dots, \, \hat\delta(\hat q_0, \mu\gamma^{2^n}).
\]
Since $\mathcal D$ has at most $2^n$ states, there must exist two integers $0 \le h < k \le 2^n$ such that $\hat\delta(\hat q_0, \mu\gamma^h) = \hat\delta(\hat q_0, \mu\gamma^k)$. Therefore $\gamma^{k-h}$ labels a cycle starting from $\hat\delta(\hat q_0, \mu\gamma^h)$. Similarly, there exist $0 \le h' < k' \le 2^n$ such that $\gamma^{k'-h'}$ labels a cycle starting from $\hat\delta(\hat q_0, \nu\gamma^{h'})$. The words 
\begin{align*}
\hat\mu &:= \mu \gamma^{h} \\
\hat\nu &:= \nu \gamma^{h'} \\
\hat\gamma &:= \gamma^{\text{lcm}(k-h, k'-h')\cdot 2^n},
\end{align*}
where the factor $2^n$ in the definition of $\hat\gamma$ ensures that $|\hat\mu|, |\hat\nu| \le |\hat\gamma|$, satisfy condition 2 of Theorem \ref{polynomialW}. 
Condition 1 of Theorem \ref{polynomialW} follows automatically from conditions 1 of this corollary. Lastly, condition 3 of Theorem \ref{polynomialW} follows from conditions 3 of this corollary and the fact that $\gamma \ndashv \mu, \nu.$
Thus we can apply Theorem \ref{polynomialW} to conclude that $\mathcal L$ is not Wheeler.
\\$(\Longrightarrow)$
Since $\mathcal L = \la D$ is not Wheeler, let $\hat\mu, \hat\nu, \hat\gamma$ be as in Theorem \ref{polynomialW}. The DFA $\mathcal D$ has at most $2^n$ states, hence the length of $\hat\gamma$ is bounded by the constant $2^{3n}+2\cdot 2^{2n}+2^n+2$. 
We have $\hat\mu\hat\gamma^* \subseteq \pf L$, so let $t_0 = q_0, t_1, \dots, t_m$ be a run of $\hat\mu\hat\gamma^n$ over $\mathcal A$. We set $u := |\hat\mu|$ and $g := |\hat\gamma|$, and 
consider the list of $n+1$ states 
\[
t_u, \; t_{u+g}, \; t_{u+2g}, \; \dots, \; t_{u+ng} = t_m 
\]
Since $\mathcal A$ has $n$ states, there must exist two integers $0 \le h < k \le n$ such that $t_{u+hg} = t_{u+kg}$. That is, there exists a state $p := t_{u+hg}$ such that $p \in \delta\left(q_0, \hat\mu\hat\gamma^h\right)$ and $\hat\gamma^{k-h}$ labels a cycle starting from $p$. We can repeat the same argument for a run of $\hat\nu\hat\gamma^n$ over $\mathcal A$ to find a state $r$ and two integers $h', k'$ such that $r \in \delta(q_0, \hat\nu\hat\gamma^{h'})$ and $\hat\gamma^{k'-h'}$ labels a cycle starting from $r$. We can then define the words
\begin{align*}
\mu &:= \hat\mu \hat\gamma^{h} \\
\nu &:= \hat\nu \hat\gamma^{h'} \\
\gamma &:= \hat\gamma^{\text{lcm}(k-h, k'-h')\cdot n}
\end{align*}
which satisfy the conditions 2 and 3.
\\Condition 4 is satisfied since $|\hat\gamma| \le 2^{3n}+2\cdot2^{2n}+2^n+2$ and $\text{lcm}(k-h,k'-h') < n^2$.
\\Finally, condition 1 is satisfied for all $i,j\ge0$. Indeed, for all $l$ the words $\hat\mu$ and $\hat\mu\hat\gamma^l$ lead to the same state of $\mathcal D$, thus $\hat\mu \equiv_\mathcal L \hat\mu\hat\gamma^l$. Similarly, for all $l$ we also have $\hat\nu \equiv_\mathcal L \hat\nu\hat\gamma^l$. Since $\forall i\; \exists s_i$ such that $\mu\gamma^i = \hat\mu \hat\gamma^{s_i}$, and similarly, $\forall j\; \exists s_j$ such that $\nu \gamma^j= \hat\nu \hat\gamma^{s_j}$, the thesis follows from $\hat\mu \not\equiv_\mathcal L \hat\nu$. 
\end{proof}

\vskip5mm

\begin{proof}[\noindent Proof of Proposition \ref{pspace}]
First of all we need to prove that the problem is in PSPACE. We will show instead that its complement is in NPSPACE, then the thesis follows from Savitch's Theorem, which states that NPSPACE = PSPACE, and the fact that PSPACE is closed under complementation. 
We  prove that checking  the conditions in Corollary \ref{nfa length} is in NPSPACE. 
 We can use non-determinism to guess, bit by bit, the length of $\mu, \nu$ and $\gamma$ and store this guessed information in three counters $u, v, g$ respectively, using $O(n)$ space for each. We also need to guess  the ending states $p,r \in Q$  of $\mu, \nu$. Then we start    guessing  the  characters of $\mu, \nu$ and $\gamma$ starting from their last one and proceeding backwards toward their first one, checking condition 3 of Corollary \ref{nfa length} in constant space. Whenever we guess a character of $\mu$ (respectively, $\nu,\gamma$) we decrease by one the counter $u$ ($v,g$), so we know when the guessing stops. If condition 3 is satisfied, we calculate (we will show later how) the sets
 $\delta(q_0, \mu)$, $\delta(q_0, \nu)$, $\delta(p, \gamma)$, and $\delta(r, \gamma)$ and check condition 2, that is, whether $p\in \delta(q_0, \mu)$, $r\in \delta(q_0, \nu)$, $p\in \delta(p, \gamma)$, and $r\in \delta(r, \gamma)$.
 If condition 2 is satisfied, we consider condition 1: for fixed $i,j\leq 2^n$ consider the automata $A^i_{\mu\gamma}$ and $A^j_{\nu\gamma}$ obtained from the NFA $\la A$ by considering as initial states the sets  $\delta(q_0, \mu\gamma^i)$,  $\delta(q_0, \nu\gamma^j)$, respectively. Notice that  we have  $\mu\gamma^i \not\equiv_{\mathcal L} \nu\gamma^j$ if and only if  $\la{A^i_{\mu\gamma}} \neq \la{A^j_{\nu\gamma}}$ and   checking whether $\la{A^i_{\mu\gamma}} = \la{A^j_{\nu\gamma}}$ can be done in polynomial space, since deciding whether two NFAs recognize the same language is a well-known PSPACE-complete problem.
 
 To conclude the proof, we claim that we are able to calculate in polynomial space, for all $q \in Q$ and for all $\beta \in \{\mu,\nu,\gamma\}^*$, the set $\delta(q,\beta)$. While guessing $\mu$ character by character, we can compute, for each state $q$, the set $Q_{\mu,q}$ of the states from which is possible to reach $q$ reading $\mu$. To build $Q_{\mu,q}$ we start from the set $Q^0 := \{q\}$ and we follow backwards the edges entering $q$ and labeled as the last character of $\mu$. We call $Q^1$ this new set of states and we repeat the process by following backwards the edges entering each state of $Q^1$ and labeled as the second to last character of $\mu$. Proceeding inductively, we compute the sets $Q^0, Q^1, \dots Q^u = Q_{\mu,q}$; notice that to calculate $Q^{k+1}$ we only need $Q^k$ and the $k$-th to last character of $\mu$, thus we can update (instead of storing) the set $Q^k$. We can do the same for $\nu$ and $\gamma$ to compute, for each $q \in Q$, the sets $Q_{\nu,q}$ and $Q_{\gamma,q}$. Once we have stored, for all $q \in Q$ and for all $\alpha \in \{\mu, \nu, \gamma\}$, the sets $Q_{\alpha,q}$, our claim follows easily. For instance, to compute $\delta(q_0, \mu\gamma)$ we build, for all $t \in Q$, the set 
 \[
 Q(t) := \{q \in Q:\; t \in \delta(q, \mu\gamma) \} = \bigcup_{p\in Q_{\gamma,t}} Q_{\mu,p}.
 \] 
 Then we have $t \in \delta(q_0, \mu\gamma)$ if an only if $q_0 \in Q(t)$.

\vskip 5mm

To prove the completeness of the problem, we will show a polynomial reduction from the universality problem for NFA, i.e. the problem of deciding whether the language accepted by a NFA $\mathcal A$ over the alphabet $\Sigma$ is such that $\la A = \Sigma^*$. 

Let $\mathcal A = (Q, q_0,\delta, F, \Sigma)$ be a NFA and let $\mathcal L = \la A$. We can assume without loss of generality that $q_0 \in F$, otherwise $\mathcal A$ would not accept the empty word and we could immediately derive that $\mathcal L \ne \Sigma^*$. Let $a,b,c$ be three characters not in $\Sigma$ and such that $a\prec b\prec c$ with respect to the lexicographical order (the order of the characters of $\Sigma$ is irrelevant in this proof). First, we build the automaton $\mathcal A'$ starting from $\mathcal A$ by adding an edge $(q_f, q_0, c)$ for each final state $q_f \in F$, see the top part of Figure \ref{pspace}. Notice that $\mathcal A'$ recognizes the language $\mathcal L' = \la{A'} = (\mathcal Lc)^* \cdot \mathcal L$, and it is straightforward to prove that $\mathcal L = \Sigma^*$ if and only if $\mathcal L' = (\Sigma + c)^*$: if $\mathcal L = \Sigma^*$, let $\alpha$ be a word in $(\Sigma + c)^*$ containing $n$ occurrences of $c$. Then $\alpha = \alpha_0\, c\, \alpha_2\, c\,\dots\, \alpha_{n-1}\, c\, \alpha_{n}$ for some $\alpha_1, \dots, \alpha_{n} \in \Sigma^*$. Hence $\alpha \in (\Sigma^*c)^*\cdot \Sigma^* = \mathcal L'$. On the other hand, if $\mathcal L \ne \Sigma^*$ let $\alpha$ be a word in $\Sigma^* \setminus \mathcal L$. Then $\alpha \cdot c \notin \mathcal L'$.

\begin{figure}
\begin{center}
\begin{tikzpicture}[->,>=stealth', semithick, initial text={}, auto, scale=.4]
\node[state, label=above:{},initial] (0) at (-6,0) {$q'_0$};

\node[state, label=above:{}, accepting] (1) at (0,5) {$q_0$};
\node[state, label=above:{}, accepting] (2) at (5,-3) {$q_1$};
\node[state, label=above:{}] (3) at (10,5) {N};
\node[state, label=above:{}, accepting] (5) at (15,5) {A};
\node[] (6) at (5,5) {$\mathcal A$};
\node (7) at (-4.5,8) {$\mathcal A'$};
\draw[dashed, rounded corners = 10] (-2,3.3) rectangle (17,6.8);
\draw[dashed, rounded corners = 10] (-3,2.3) rectangle (18,9.8);

\draw (0) edge node {$a$} (1);
\draw (0) edge[below] node {$b$} (2);

\draw (1) edge[loop above] node {c} (1);

\draw (5) edge[bend right = 40, above] node {$c$} (1);
\draw (2) edge[loop right] node {$\Sigma, c$} (2);

\end{tikzpicture}
\end{center}
\caption{The automaton $\mathcal A''$. Every accepting state of $\mathcal A$, labeled A in the figure, has a back edge labeled $c$ connecting it to $q_0$. Conversly, non-accepting states of $\mathcal A$, labeled N in the figure, do not have such back edges.}
\label{fig:pspace}
\end{figure}
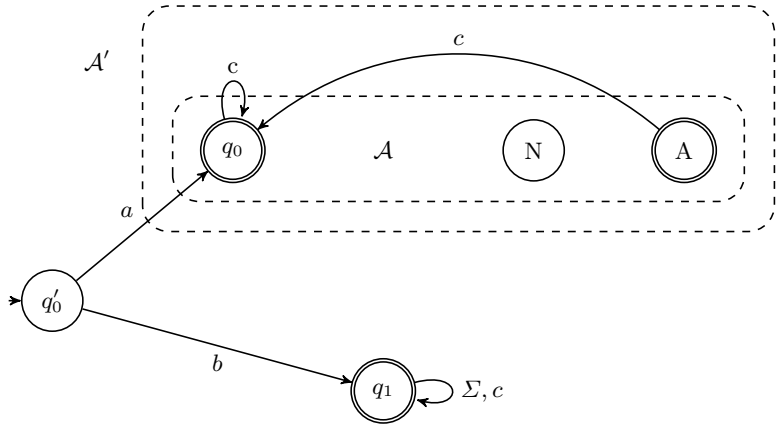

We build a second automaton $\mathcal A''$ as depicted in Figure \ref{fig:pspace}. Let $\mathcal L'' = \la {A''}$ be the language recognized by $\mathcal A''$. We claim that $\mathcal L = \Sigma^*$ if and only if $\mathcal L''$ is Wheeler.
\\$(\Longrightarrow)$ If $\mathcal L = \Sigma^*$, we have already proved that $\mathcal L' = (\Sigma+c)^*$. Hence we have $\mathcal L'' = (a + b) \cdot (\Sigma+c)^*$. The minimum DFA recognizing $\mathcal L''$ has only one loop, therefore by Theorem \ref{polynomialW} $\mathcal L''$ is Wheeler.
\\$(\Longleftarrow)$ If $\mathcal L \ne \Sigma^*$, let $\alpha$ be a word in $\Sigma^* \setminus \mathcal L$. Notice that $\alpha \ne \varepsilon$ since we assumed that $\varepsilon \in \mathcal L$. Every possible run of $\alpha$ over $\mathcal A$ must lead to a non-accepting state, hence $\alpha \cdot c \notin \mathcal L'$. 
This implies that for all $i \ge 0$ we have $a \cdot c^i \cdot \alpha \cdot c \notin \mathcal L''$ (notice that the only edge labeled $c$ leaving $q_0$ ends in $q_0$). On the other hand, for all $j \ge 0$ we have $bc^j \cdot \alpha \cdot c \in \mathcal L''$, hence for all $i, j \ge 0$ we have $ac^i \not\equiv_{\mathcal L''} bc^j$. 
Thus the following monotone sequence in $\pf {L''}$
\[
ac \prec bc \prec acc \prec bcc \prec \dots \prec ac^n \prec bc^n \prec \dots
\]
is not eventually constant modulo $\equiv_\mathcal{L''}$. From Lemma \ref{monotone} it follows that $\mathcal L''$ is not Wheeler.
\end{proof}

\vskip5mm

\begin{proof}[\noindent Proof of Lemma \ref{short}]
Suppose by contradiction that there exists a class $C_i$ such that for all $\alpha \in C_i$ it holds $|\alpha| \ge n + n^2$, and let $\alpha \in C_i$ be a word of minimum length. Consider the first $n+1$ states $q_0=t_0, ..., t_n$ of $\mathcal A$ visited by reading the first $n$ characters of $\alpha$. Since $\mathcal A$ has only $n$ states, there must exist $0 \le i,j \le n$ with $i < j$ such that $t_i=t_j$. Let $\alpha'$ be the prefix of $\alpha$ of length $i$ (if $i=0$ then $\alpha' = \varepsilon$), let $\delta$ be the factor of $\alpha$ of length $j-i$ labeling the path $t_i,...,t_j$, and let $\zeta$ be the suffix of $\alpha$ such that $\alpha = \alpha' \delta \zeta$. By construction, the words $\alpha$ and $\beta := \alpha' \zeta$ end in the same state, hence $\alpha \equiv_\mathcal L \beta$. Moreover, from $|\beta| < |\alpha|$ and the minimality of $\alpha$ it follows that $\alpha \not\equiv_\mathcal L^c \beta$. 
\\Suppose that $\alpha \prec \beta$, the other case being completely symmetrical. Since $\alpha$ and $\beta$ share the same suffix $\zeta$, they end with the same character. This means that the words $\alpha$ and $\beta$, which are Myhill-Nerode equivalent but not $\equiv_\mathcal L^c$ equivalent, were not split into two distinct $\equiv_\mathcal L^c$-classes due to input-consistency, therefore there must exists a word $\eta$ such that $\alpha \prec \eta \prec \beta$ and $\eta \not\equiv_\mathcal L \alpha$. Formally, assume by contradiction that for all words $\eta$ such that $\alpha \prec \eta \prec \beta$ it holds $\eta \equiv_\mathcal L \alpha$. Then, by definition of $\equiv_\mathcal L^c$, it would follow $\alpha \equiv_\mathcal L^c \beta$, a contradiction.
\\Let $\eta$ be a word such that $\alpha \prec \eta \prec \beta$ and $\eta \not\equiv_\mathcal L \alpha$. From $\zeta \dashv \alpha, \beta$ it follows that $\zeta \dashv \eta$, so we can write $\eta = \eta' \zeta$ for some $\eta' \in \Sigma^*$. Recall that by construction $\alpha = \alpha' \delta \zeta$ with $|\alpha' \delta| \le n$, hence $|\zeta| \ge n^2$. Consider the last $n^2+1$ states $r_0, ..., r_{n^2}$ of $\mathcal A$ visited by reading the word $\alpha$, and the last $n^2+1$ states $p_0, ..., p_{n^2}$ visited by reading the word $\eta$. Since $\mathcal A$ has only $n$ states, there must exist $0 \le i,j \le n^2$ with $i < j$ such that $(r_i, p_i) = (r_j, p_j)$. Notice that it can't be $r_i = p_i$, otherwise from the determinism of $\mathcal A$ it would follow $r_{n^2} = p_{n^2}$; from the minimality of $\mathcal A$ it would then follow $\alpha \equiv_\mathcal L \eta$, a contradiction.
\\Let $\zeta''$ be the suffix of $\zeta$ of length $n^2-j$, and let $\gamma$ be the factor of $\zeta$ of length $j-i$ labeling the path $r_i,...,r_j$. Since $|\zeta| \ge n^2$, there exists $\zeta' \in \Sigma^*$ such that $\zeta = \zeta' \gamma \zeta''$. We can then rewrite $\alpha, \eta$ and $\beta$ as
\begin{align*}
    \alpha &= \alpha' \delta \zeta = \alpha' \delta \zeta' \gamma \zeta'' \\
    \eta &= \eta' \zeta = \eta' \zeta' \gamma \zeta'' \\
    \beta &= \alpha' \zeta = \alpha' \zeta' \gamma \zeta''.
\end{align*}
Let $k$ be an integer such that $|\gamma^k|$ is greater than $|\alpha' \delta \zeta'|$ and $|\eta' \zeta'|$. Set $\mu := \eta' \zeta'$; from $\alpha \prec \eta \prec \beta$ it follows that $\alpha' \delta \zeta' \prec \mu \prec \alpha' \zeta'$. If $\gamma^k \prec \mu$ set $\nu := \alpha' \zeta'$, otherwise set $\nu := \alpha' \delta \zeta'$. In both cases, the hypothesis of Theorem \ref{polynomialW} are satisfied, since $\gamma^k$ labels two cycles starting from the states $r_i$ and $p_i$, that we have proved to be distinct. We can conclude that $\mathcal L$ is not Wheeler, a contradiction, and the thesis follows. 
\end{proof}

\vskip5mm

\begin{proof}[\noindent Proof of Proposition \ref{DFA to WDFA}]
Consider the pairwise distinct equivalence classes $C_1, \dots, C_m$ of the equivalence $\equiv_\mathcal L^c$. Clearly, the minimum Wheeler automaton recognizing $\mathcal L$ has $m$ states. We can assume without loss of generality that the equivalence classes are co-lexicographically ordered, i.e. $C_i \prec C_{i+1}$ for all $i<m$. For sake of simplicity, given a word $\alpha \in \Sigma^*$ we will write $[\alpha]$ to indicate its equivalence class modulo $\equiv_\mathcal L$, and we will write $[\alpha]_W$ to indicate its equivalence class modulo $\equiv_\mathcal L^c$. 
\\Let $K := \{ 1, ..., k \}$ and, for all $i$, let $K_i := \{ j\in K:\; [\alpha_j]_W = C_i \}$. Since each $C_i$ is convex in $\pf L$, each $K_i$ must be convex in $K$, that is each $K_i$ is an interval.  
Therefore the list of equivalence classes $[\alpha_1]_W, ..., [\alpha_k]_W$ must be partitioned in consecutive runs of the same class, each class appearing in one and only one run. From Lemma \ref{short} we know that each equivalence class has at least one representative in $\pf L^{\le d}$, hence the list $[\alpha_1]_W, ..., [\alpha_k]_W$ must contain exactly $m$ runs. 

For all $1 \le j < k$, we have $[\alpha_j]_W \ne [\alpha_{j+1}]_W$ if and only if
\[\Big( [\alpha_j] \ne [\alpha_{j+1}] \Big) \vee \Big( [\alpha_j] = [\alpha_{j+1}] \wedge \text{last}(\alpha_j) \ne \text{last}(\alpha_{j+1}) \Big), \]
where $\text{last}(\alpha)$ denotes the last character of $\alpha$. This means that we can identify the $m$ runs of the equivalence $\equiv_\mathcal L^c$ just by looking at the two lists $\alpha_1, ..., \alpha_k$ and $[\alpha_1], ..., [\alpha_k]$: whenever $\text{last}(\alpha_j) \ne \text{last}(\alpha_{j+1})$ or $[\alpha_j] \ne [\alpha_{j+1}]$, we know that a new $\equiv_\mathcal L^c$ run must start at $\alpha_{j+1}$. 
\\In $O(k)$ we are able to determine the $m$ runs and to pick a representative for each of them, i.e. we can find $m$ indexes $i_1, ..., i_m$ such that for all $1 \le j \le m$ it holds $\alpha_{i_j} \in C_i$. We call the set $\{ a_{i_1}, ..., a_{i_m} \}$ a \textit{fingerprint} of the language $L$, i.e. a set of words that has cardinality $m$ such that distinct elements of the set belong to distinct $\equiv_\mathcal L^c$-classes.

We show how to build the minimum Wheeler DFA recognizing $\mathcal L$, starting from any fingerprint of $\mathcal L$ and the standard minimum DFA recognizing $\mathcal L$. Let $\{ \beta_1, ..., \beta_m \}$ be a fingerprint of $\mathcal L$ and let $\mathcal A$ be the minimum DFA recognizing $\mathcal L$. We can assume without loss of generality that $\beta_1 \prec ... \prec \beta_m$. We build the automaton $\mathcal A^W=(Q, \beta_1, \delta, F,\Sigma)$, where the set of states is $Q = \{ \beta_1, ..., \beta_m \}$ and the set of final states is $F = \{ \beta_j:\; \beta_j \in \mathcal L \}$. The transition function $\delta$ can be computed as follow. For all $1 \le j \le m$ and for all $c \in \Sigma$, check whether $\beta_j \cdot c \in \pf L$. If $\beta_j \cdot c \notin \pf L$, there are no edges labeled $c$ that exit from $\beta_j$. If instead $\beta_j \cdot c \in \pf L$, locate $\beta_j \cdot c$ using a binary search. There are three possible cases.
\begin{enumerate}
    \item $\beta_j \cdot c \prec \beta_1$. Then $\delta(\beta_j, c) = \beta_1$.
\item $\beta_m \prec \beta_j \cdot c$. Then $\delta(\beta_j, c) = \beta_m$.
\item There exists $s$ such that $\beta_s \preceq \beta_j \cdot c \preceq \beta_{s+1}$. It can not be the case that both $\beta_jc \not\equiv_\mathcal L \beta_s$ and $\beta_jc \not\equiv_\mathcal L \beta_{s+1}$, since $\{ \beta_1, ..., \beta_m \}$ is a fingerprint of $\mathcal L$. Hence we distinguish three cases.
\begin{enumerate}
    \item $\beta_s \equiv_\mathcal L \beta_jc \not \equiv_\mathcal L \beta_{s+1}$. Then $\delta(\beta_j, c) = \beta_s$.
    \item $\beta_s \not\equiv_\mathcal L \beta_jc \equiv_\mathcal L \beta_{s+1}$. Then $\delta(\beta_j, c) = \beta_{s+1}$.
    \item $\beta_s \equiv_\mathcal L \beta_jc \equiv_\mathcal L \beta_{s+1}$. Since $\{ \beta_1, ..., \beta_m \}$ is a fingerprint of $\mathcal L$, it is either $c = \text{last}(\beta_jc) = \text{last}(\beta_s)$, in which case $\delta(\beta_j, c) = \beta_s$, or $c = \text{last}(\beta_{s+1})$, in which case $\delta(\beta_j, c) = \beta_{s+1}$.
\end{enumerate}
\end{enumerate}
\end{proof}

\vskip5mm

\begin{proof}[\noindent Proof of Proposition \ref{W to GW}]
The problem is in NP, since we can use non-determinism to guess the order of the alphabet and then check whether such order makes the NFA Wheeler.

To prove the hardness, we show a polynomial reduction from the problem of deciding whether a NFA is Wheeler. Let $\mathcal A$ be a NFA with initial state $q_0$, over the alphabet $\Sigma=\{a_1\dots,a_\sigma\}$ ordered by the relation $a_1\prec\dots\prec a_\sigma$. We want to build a new automaton $\mathcal A'$ such that $\mathcal A'$ is a GWNFA if and only if $\mathcal A$ is Wheeler. $\mathcal{A}'$ will be an automaton of size $|\mathcal{A}|+O(\sigma)$ over the alphabet $\Sigma'$ of size $O(\sigma)$. 

The automaton $\mathcal A'$ will be built starting from $\mathcal A$ and adding extra states and transitions. We define the new alphabet as
\[ \Sigma' = \{ a_1, \dots , a_\sigma, x_1, \dots, x_{\sigma-1}, e, f \}, \]
with $x_i, e, f \notin \Sigma $, and we add two final states $q_e$ and $q_f$. We then build $\sigma-1$ gadgets, one for each pair of consecutive characters $(a_i, a_{i+1})$ of $\Sigma$, each one connected to $\mathcal A \cup \{ q_e, q_f \}$ as depicted in Figure \ref{f2}. This completes the construction of the automaton $\mathcal A'$. Notice that $\mathcal A$ is input-consistent, but in general it is not deterministic.

\begin{figure}
\begin{center}
\begin{tikzpicture}[->,>=stealth', semithick, initial text={}, auto, scale=.4]
\node[state, label=above:{},initial] (0) at (-6,0) {$q_0$};

\node[state, label=above:{}] (1) at (0,3) {$q_i^3$};

\node[state, label=above:{}] (3) at (7,3) {$q_{i}^5$};
\node[state, label=above:{}] (4) at (10,7) {$q_{i}^7$};
\node[state, label=above:{}] (5) at (4,7) {$q_{i}^2$};

\node[state, label=above:{}] (6) at (7,-3) {$q_{i}^4$};
\node[state, label=above:{}] (7) at (10,-7) {$q_{i}^6$};
\node[state, label=above:{}] (8) at (4,-7) {$q_{i}^1$};

\node[state, label=above:{}, accepting] (9) at (14,3) {$q_e$};
\node[state, label=above:{}, accepting] (10) at (14,-3) {$q_f$};

\draw (0) edge node {$a_{i+1}$} (1);
\draw (0) edge[below, bend right=20] node {$x_i$} (6);

\draw (1) edge node {$x_i$} (3);

\draw (3) edge[below] node[ xshift=8pt] {$x_i$} (4);
\draw (4) edge[above] node {$a_i$} (5);
\draw (5) edge node {$x_i$} (3);
\draw (3) edge node {$e$} (9);

\draw (6) edge[below, pos=.2] node[xshift=16pt] {$x_i$} (7);
\draw (7) edge[above] node {$a_i$} (8);
\draw (8) edge node {$x_i$} (6);
\draw (6) edge node {$f$} (10);

\end{tikzpicture}
\end{center}
\caption{The gadget $G_i$, connected to $q_0$ and to the sinks $q_e$ and $q_f$.}
\label{f2}
\end{figure}
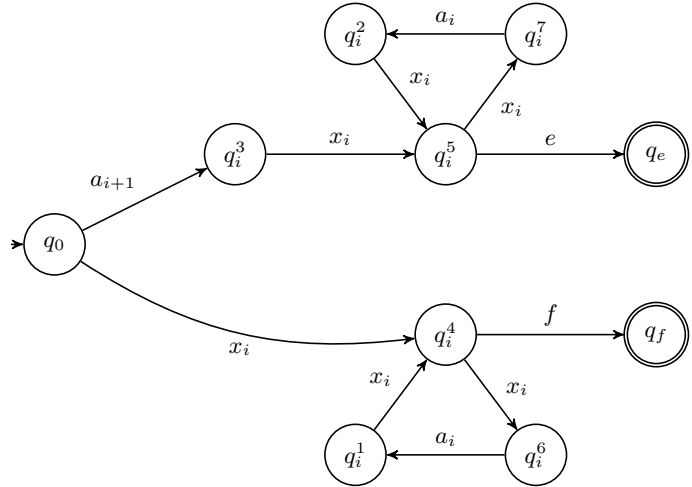

We want to show that $\mathcal{A}$ is Wheeler according to the order $( \Sigma, \prec )$ if and only if $\mathcal{A'}$ is a GWNFA.\\
($\Longrightarrow$) Define the order $\prec'$ over $\Sigma'$ by setting
\[
a_1 \prec' ... \prec' a_\sigma \prec' x_1 \prec' ... \prec' x_{\sigma-1} \prec' e \prec' f.\]
We show that $\mathcal A'$ is Wheeler according to $(\Sigma' , \prec')$ by ordering its states. Since $\mathcal{A}$ is Wheeler, there already exists an order of its states that makes $\mathcal{A}$ Wheeler. 
Therefore, we simply need to extend this order to the states of $\mathcal A'$. 
Recall that in \cite{ADPP} it has been proved that a WNFA is Wheeler if and only if, for each pair of states $q,p$ such that $I_q \ne I_p$, either $I_q \preceq I_p$ or $I_p \preceq I_q$ holds, where by definition $I_q \preceq I_p$ if and only if, for all $\alpha \in I_q$ and for all $\beta \in I_p$ such that $\{\alpha, \beta\}\not\subseteq I_q\cap I_p$, we have $\alpha \prec \beta$. Therefore we check, for each pair of states $q$ and $p$ in $\mathcal A'$, that either $I_q \preceq I_p$ (implying $q<p$) or $I_p \preceq I_q$ (implying $p<q$) holds.
Note that when $I_q$ and $I_p$ are disjoint, the condition $I_q \preceq I_p$ translates to the following: for all $\alpha \in I_q$ and for all $\beta \in I_p$, $\alpha \prec \beta$ . 
In the discussion that follows we never compare two states that belong to $\mathcal A$, since the order between them is already established.

First of all, we will order, for each $i$, the states with incoming edge $x_i$. Since $x_i \notin \Sigma$, the only states to compare are the one belonging to the gadget $G_i$, i.e. $q_i^4, q_i^5, q_i^6, q_i^7$. Consider the languages
\begin{align*}
    I_4:=I_{q_{i}^4}&=x_i(x_ia_ix_i)^*\\
    I_5:=I_{q_{i}^5}&=a_{i+1}x_i(x_ia_ix_i)^*\\
    I_6:=I_{q_{i}^6}&=x_ix_i(a_ix_ix_i)^*\\
    I_7:=I_{q_{i}^7}&=a_{i+1}x_ix_i(a_ix_ix_i)^*.
\end{align*}
Since $a_{i+1}\prec'x_i$ we have $I_4,I_5 \prec I_6,I_7$. Moreover, consider any two words $\alpha \in I_4$ and $\beta \in I_5$. 
If $|\alpha|<|\beta|$, we have $\alpha \dashv \beta$ hence $\alpha \prec \beta$. If $|\alpha|\ge|\beta|$ instead, then the last $|\beta|$ characters of $\alpha$ must be $a_ix_i(x_ia_ix_i)^m$, where $m$ is such that $\beta=a_{i+1}x_i(x_ia_ix_i)^m$. 
From $a_i \prec ' a_{i+1}$ we still get $\alpha \prec \beta$, hence $I_4 \prec I_5$. Similarly we can prove that $I_6 \prec I_7$. It immediately follows that we need to order the states as follows: $q_i^4 < q_i^5 < q_i^6 < q_i^7$.

Secondly, for each $1 \le i \le \sigma$ we need to sort the states with incoming edges labeled $a_i$. Note that the automaton $\mathcal A$ might contain states with incoming label $a_i$ and we need to consider such states as well. For $i=\sigma$, the task is easy. The only gadget with a state labeled $a_\sigma$ is $G_{\sigma-1}$, and such state is $q_{\sigma-1}^3$. 
Notice that $I_{q_{\sigma-1}^3}=\{a_{\sigma}\}$. 
Let $q$ be a state of $\mathcal A$ with $\lambda(q) = a_\sigma$. Since every word in $I_q$ ends with $a_\sigma$, it trivially follows that $I_{q_{\sigma-1}^3} \preceq I_q$. If $I_q \ne \{a_\sigma\}$, we are forced to set $q_{\sigma-1}^3 < q$. If instead $I_q = \{a_\sigma\}$, the order of $q$ and $q_{\sigma-1}^3$ does not matter. For sake of consistency, we set $q_{\sigma-1}^3 < q$. \\
For $i=1$, the only gadget with states labeled $a_1$ is $G_{1}$, and such states are $q_1^1$ and $q_1^2$. 
Let $q$ be a state of $\mathcal A$ with $\lambda(q) = a_1$ 
and consider the languages
\begin{align*}
    I_1:=I_{q_{1}^1}&=x_1x_1a_1(x_1x_1a_1)^*\\
    I_2:=I_{q_{1}^2}&=a_2x_1x_1a_1(x_1x_1a_1)^*.
\end{align*}
For all $\alpha \in I_q$, we have either $\alpha = a_1$ or $\alpha = \alpha' a_j a_1$ for some $\alpha' \in \Sigma^*$ and some $1 \le j \le \sigma$. 
In both cases, $\alpha$ must precede co-lexicographically every word of $I_1$ and $I_2$, thus $I_q \prec I_1, I_2$. To compare $I_1$ and $I_2$, consider any two words $\alpha \in I_1$ and $\beta \in I_2$. 
If $|\alpha|<|\beta|$, we have $\alpha \dashv \beta$ hence $\alpha \prec \beta$. If $|\alpha|\ge|\beta|$ instead, then the last $|\beta|$ characters of $\alpha$ must be $a_1x_1x_1a_1(x_1x_1a_1)^m$, where $m$ is such that $\beta=a_2x_1x_1a_1(x_1x_1a_1)^m$. 
From $a_1 < ' a_2$ we still get $\alpha \prec \beta$, hence $I_1 \prec I_2$. 
It follows that $q_1^1$ and $q_1^2$ must follow $q$ (and all the other states of $\mathcal A$ with label $a_1$), and we must set $q_1^1 < q_1^2$. \\
Lastly, if $1 < i < \sigma$ we should compare the sates $q_{i-1}^3, \; q_i^1$ and $q_i^2$ with the sates in $Q_i:= \{ q \in \mathcal A: \; \lambda(q) = a_i \} $. Applying both of the reasoning discussed in the cases $i=1$ and $i=\sigma$, we can conclude that the state $q_{i-1}^3$ precedes all the states in $Q_i$ and that the states $q_i^1$ and $q_i^2$ follow all the states in $Q_i$ and must be ordered as $q_i^1 < q_i^2$.

The order of the states of $\mathcal A'$ that we described makes $\mathcal A'$ Wheeler with respect to $(\Sigma' , \prec' )$, hence making it a GWNFA.\\ 
\noindent($\Longleftarrow$) If $\mathcal A'$ is a GWNFA, then there exists an order $\prec'$ over $\Sigma'$ that makes $\mathcal A'$ Wheeler. Since $\mathcal A$ is a sub-automaton of $\mathcal A'$, it follows that even $\mathcal A$ is Wheeler according to $\prec'$. 
Let $\widetilde \prec$ be the restriction of $\prec'$ over the alphabet $\Sigma$; we want to show that $\widetilde \prec$ is the same order as $\prec$. Assume by contradiction that $\widetilde \prec \ne \prec$. 
If for all $1 \le i < \sigma$ we have $a_i \widetilde \prec a_{i+1}$, then $\widetilde \prec = \prec$, a contradiction. Hence there exists $1 \le i < \sigma$ such that $a_{i+1} \widetilde \prec a_{i}$. Since $\prec'$ extends $\widetilde \prec$, this implies that $a_{i+1} \prec' a_{i}$. We will show that $\mathcal A'$ is not Wheeler according to $\prec'$, a contradiction. 
Define the words $\mu:=x_i,$ $\nu:=a_{i+1}x_i$ and $\gamma:=x_ia_ix_i$. From $\mu \dashv \gamma$ and $a_{i+1} \prec' a_{i}$ we have $\mu, \nu \prec \gamma$. The word $\gamma$ labels two cycles in $\mathcal A'$ starting from two distinct states, i.e. $q_i^4$ and $q_i^5$. 
Moreover, $\mu$ and $\nu$ label two paths that start from the initial state $q_0$ and end in $q_i^4$ and $q_i^5$ respectively. Since $q_i^4$ and $q_i^5$ are not Myhill-Nerode equivalent, we can apply Theorem \ref{polynomialW} to conclude that $\mathcal A'$ is not Wheeler according to $\prec'$, a contradiction. 
Thus $\widetilde \prec$ and $\prec$ coincide.\\
We have shown that $\mathcal A$ is Wheeler according to $\prec'$, and that $\prec'$ extends $\prec$. Therefore we can conclude that $\mathcal A$ is Wheeler according to $\prec$.
\end{proof}

\vskip5mm

\begin{definition}[Betweenness]
Input: a list $Y$ of $n$ distinct elements $Y=y_1,...,y_n$ and $k<n^3$ ordered triples $(a_1,b_1,c_1),..., (a_k,b_k,c_k)$, where each element of each triple belongs to $Y$. Elements belonging to the same triple are distinct.\\
Output: yes/no answer. The answer is ``yes'' if and only if there exists a total order < of $Y$ such that, for each $k$, either $a_k<b_k<c_k$ or $a_k>b_k>c_k$.
\end{definition}

\begin{proof}[\noindent Proof of Proposition \ref{GWDFA}]
We can prove that both problems are in NP using an argument similar to the one employed in the proof of Proposition \ref{W to GW}.

To prove the hardness, we show a polynomial reduction from the betweenness problem to both of the problems described; we will use exactly the same reduction for both problems. We start from an instance $I=(Y,K)$ of the betweenness problem, where $Y$ is the set $Y=\{y_1\dots,y_n\}$ and $K\subseteq \mathcal P(T^3)$ is a collection of $k$ triples $(a_1,b_1,c_1),\dots, (a_k,b_k,c_k)$, for some $1<k<n^3$. We build a DFA $\mathcal A$ of size $O(n+k)$, over an alphabet of size $O(n+k)$. The alphabet is $\Sigma=Y\cup\{x_1\dots,x_k,e,f\}$, where we introduce a new character $x_i$ for each triple $(a_i,b_i,c_i)\in K$ and two extra ``ending'' characters $e$ and $f$. To build $G$, we start with the initial state $q_0$ connected with $n$ states $q_1, \dots,q_n$ through the edges $(q_0,y_j,q_j)$ for each $1\le j\le n$. We also add two sinks $q_e$ and $q_f$, the only final states. 
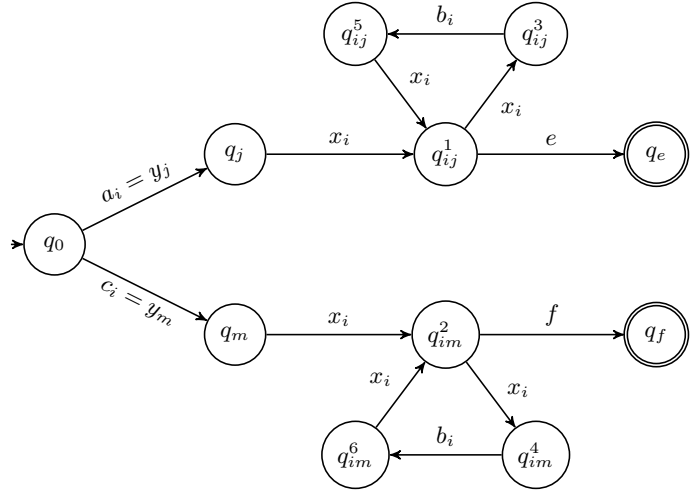
\begin{figure}
\begin{center}
\begin{tikzpicture}[->,>=stealth', semithick, initial text={}, auto, scale=.4]
\node[state, label=above:{},initial] (0) at (-6,0) {$q_0$};

\node[state, label=above:{}] (1) at (0,3) {$q_j$};
\node[state, label=above:{}] (2) at (0,-3) {$q_m$};

\node[state, label=above:{}] (3) at (7,3) {$q_{ij}^1$};
\node[state, label=above:{}] (4) at (10,7) {$q_{ij}^3$};
\node[state, label=above:{}] (5) at (4,7) {$q_{ij}^5$};

\node[state, label=above:{}] (6) at (7,-3) {$q_{im}^2$};
\node[state, label=above:{}] (7) at (10,-7) {$q_{im}^4$};
\node[state, label=above:{}] (8) at (4,-7) {$q_{im}^6$};

\node[state, label=above:{}, accepting] (9) at (14,3) {$q_e$};
\node[state, label=above:{}, accepting] (10) at (14,-3) {$q_f$};

\draw (0) edge[sloped] node {$a_i=y_j$} (1);
\draw (0) edge[sloped,below] node {$c_i=y_m$} (2);

\draw (1) edge node {$x_i$} (3);
\draw (2) edge node {$x_i$} (6);

\draw (3) edge[below] node[ xshift=8pt] {$x_i$} (4);
\draw (4) edge[above] node {$b_i$} (5);
\draw (5) edge node {$x_i$} (3);
\draw (3) edge node {$e$} (9);

\draw (6) edge[below, pos=.2] node[xshift=16pt] {$x_i$} (7);
\draw (7) edge[above] node {$b_i$} (8);
\draw (8) edge node {$x_i$} (6);
\draw (6) edge node {$f$} (10);

\end{tikzpicture}
\end{center}
\caption{The gadget $G_i$ related to the $i$-th triple.}
\label{f1}
\end{figure}
We add the states $q^1_{ij},q^3_{ij},q^5_{ij}$ (see Figure \ref{f1}) and the transitions \[ \delta(q_j,x_i) = q^1_{ij},\quad \delta(q^1_{ij},x_i) = q^3_{ij},\quad \delta(q^3_{ij},b_i) = q^5_{ij},\quad \delta(q^5_{ij},x_i) = q^1_{ij} \]
and $\delta(q^1_{ij},e) = q_e$. We repeat the same process with $c_i$: given the integer $m$ such that $c_i=y_{m}$, we add the states $q^2_{im},q^4_{im},q^6_{im}$ and the transitions \[ \delta(q_m,x_i)=q^2_{im},\quad \delta(q^2_{im},x_i)=q^4_{im},\quad\delta(q^4_{im},b_i)=q^6_{im},\quad\delta(q^6_{im},x_i)=q^2_{im} \]
and $\delta(q^2_{im},f)=q_f.$ Lastly, we remove the states among $q_1, \dots,q_n$ that don't have outgoing edges. More formally, we define the sets $A:=\{a_1,\dots,a_k\}$ and $C:=\{c_1,\dots,c_k\}$ and we remove from $G$ all the states $q_j$ such that $y_j\notin A\cup C$. 
%
%
We show that the instance $I=(Y,K)$ of the betweennes problem is satisfiable if and only if $\mathcal A$ is a GWNFA, if and only if $\la A$ is GW.\\
($\Longrightarrow$) Since $I=(Y,K)$ is satisfiable, there exists an ordering $\pi:Y\rightarrow \{1,...,n\}$ of the elements of $Y$ satisfying $I$. 
We order $\Sigma$ as follows: \[
\pi^{-1}(1)\prec...\prec\pi^{-1}(n)\prec x_1\prec\dots\prec x_k\prec e\prec f.
\]
This ordering induce a partial order on the states of $\mathcal A$, where states with different incoming labels are ordered by such labels. Therefore we only need to order the states of $\mathcal A$ with the same incoming label. \\For each $1\le i\le k$, the only states of $\mathcal A$ with incoming label $x_i$ are $q_{ij}^1,q_{ij}^3,q_{im}^2,q_{im}^4$, where $j$ and $m$ are integers such that $a_i=y_j$ and $c_i=y_m$. Since, by construction, the order $\pi$ satisfies the instance $I$, then only two cases can occur: either $\pi(a_i)<\pi(b_i)<\pi(c_i)$, or $\pi(c_i)<\pi(b_i)<\pi(a_i)$. In the first case, we set $q_{ij}^1<q_{im}^2<q_{ij}^3<q_{im}^4$. To realize that this is in fact the correct order of the states, consider the following languages:
\begin{align*}
    I_1&:=I_{q_{ij}^1}=\{\alpha\in\Sigma^*:\;\delta(q_0,\alpha)=q_{ij}^1\}=a_ix_i(x_ib_ix_i)^*\\
    I_2&:=I_{q_{im}^2}=c_ix_i(x_ib_ix_i)^*\\
    I_3&:=I_{q_{ij}^3}=a_ix_ix_i(b_ix_ix_i)^*\\
    I_4&:=I_{q_{im}^4}=c_ix_ix_i(b_ix_ix_i)^*.
\end{align*}
Since, by construction, we have $a_i,b_i,c_i\prec x_i$, it follows that $I_1,I_2\prec I_3,I_4$. Moreover, from $\pi(a_i)< \pi(b_i)< \pi(c_i)$ we also have that $I_1\prec I_2$ and $I_3\prec I_4$, which completes the ordering. Symmetrically, if $\pi(c_i)<\pi(b_i)<\pi(a_i)$ then we set $q_{im}^2<q_{ij}^1<q_{im}^4<q_{ij}^3$. \\
We still need to order the states of $\mathcal A$ whose incoming labels belong to $Y$. For each $1\le p\le n$, the states with incoming label $y_p$ belong to the sets $V_5:=\{q_{ij}^5:\;b_i=y_p\}$ or $V_6:=\{q_{im}^6:\;b_i=y_p\}$ or $V_p$, where $V_p=\{q_p\}$ if $q_p$ is a state of $\mathcal A$ (i.e. if $y_p\in A\cup C$) and $V_p=\emptyset$ otherwise. If $V_p=\{q_p\}$, we set $q_p$ as the smallest state. We then sort the states of $V_5$ and $V_6$ by their first subscript; when the first subscript is equal, i.e. the two states that we want to confront are $q_{ij}^5$ and $q_{im}^6$ (with $a_i=y_j$ and $c_i=y_m$), then we set $q_{ij}^5<q_{im}^6$ if $\pi(a_i)< \pi(b_i)<\pi(c_i)$ and we set $q_{im}^6<q_{ij}^5$ if $\pi(c_i)<\pi(b_i)<\pi(a_i)$. This can be deduced by confronting the following languages:
\begin{align*}
    I_5:=I_{q_{ij}^5}&=a_ix_ix_ib_i(x_ix_ib_i)^*\\
    I_6:=I_{q_{i'm}^6}&=c_{i'}x_{i'}x_{i'}b_{i'}(x_{i'}x_{i'}b_{i'})^*.
\end{align*}
If $i< i'$ then, by construction, we have $x_i\prec x_{i'}$, hence $I_5\prec I_6$. Symmetrically, if $i'<i$ then we have $I_6\prec I_5$. Lastly, if $i=i'$ then the order between $I_5$ and $I_6$ is determined solely by $a_i$ and $c_{i'}=c_i$. 

Since there is only one state with incoming label $e$ and only one state with incoming label $f$, we have finished. The order described makes $\mathcal A$ Wheeler, thus $\mathcal A$ is a GWNFA and $\la A$ is GW. 
\\($\Longleftarrow$) Assume that the instance $I=(Y,K)$ of the betweenness problem is unsatisfiable. We prove that $\la A$ is not GW. Assume by contradiction that $\la A$ is GW, then there exists an ordering $\pi'$ of the elements of $\Sigma$ such that $\la A$ is Wheeler according to said order. Recall that $Y\subseteq\Sigma$ and consider the order $\pi:=\left. \pi' \right|_{Y}$. Since $(Y,K)$ is unsatisfiable, $\pi$ must violate one of the constraints, i.e. there exists an $1\le i\le k$ such that either $\pi(a_i),\pi(c_i)<\pi(b_i)$ or $\pi(b_i)<\pi(a_i),\pi(c_i)$. Define the words $\mu:=a_ix_i,\; \nu:=c_ix_i$ and $\gamma:=x_ib_ix_i$; then it is either $\mu,\nu\prec\gamma$ or $\gamma\prec\mu,\nu$ (here the co-lexicographic order $\prec$ is calculated with respect to $\pi'$). By construction, $\gamma$ labels two cycles in $\mathcal A$ 
starting from two distinct states, $q_{ij}^1$ and $q_{im}^2$, which are not Myhill-Nerode equivalent. Moreover, $\mu$ and $\nu$ label two paths that start from the initial state $q_0$ and end in $q_{ij}^1$ and $q_{im}^2$ respectively. We can then apply the Theorem \ref{polynomialW} to conclude that $\la A$ is not Wheeler according to $\pi'$, a contradiction. Therefore $\la A$ is not GW, which automatically implies that $\mathcal A$ is not a GWNFA.
\end{proof}

\bibliographystyle{splncs04}
\bibliography{bibliography2}

\begin{thebibliography}{10}
\providecommand{\url}[1]{\texttt{#1}}
\providecommand{\urlprefix}{URL }
\providecommand{\doi}[1]{https://doi.org/#1}

\bibitem{ADPP}
Alanko, J., D'Agostino, G., Policriti, A., Prezza, N.: Regular languages meet
  prefix sorting. In: Proceedings of the 2020 ACM-SIAM Symposium on Discrete
  Algorithms. pp. 911--930 (2020). \doi{10.1137/1.9781611975994.55},
  \url{https://epubs.siam.org/doi/abs/10.1137/1.9781611975994.55}

\bibitem{ADPP2}
Alanko, J., D'Agostino, G., Policriti, A., Prezza, N.: {Wheeler Languages}.
  CoRR arXiv:2002.10303 (Feb 2020)

\bibitem{alanko2019tunneling}
Alanko, J., Gagie, T., Navarro, G., Seelbach~Benkner, L.: Tunneling on wheeler
  graphs. In: 2019 Data Compression Conference (DCC). pp. 122--131 (2019).
  \doi{10.1109/DCC.2019.00020}

\bibitem{backurs2016regular}
Backurs, A., Indyk, P.: Which regular expression patterns are hard to match?
  In: 2016 IEEE 57th Annual Symposium on Foundations of Computer Science
  (FOCS). pp. 457--466 (2016). \doi{10.1109/FOCS.2016.56}

\bibitem{equi2020complexity}
Equi, M., Grossi, R., Makinen, V.: {On the Complexity of Exact Pattern Matching
  in Graphs: Binary Strings and Bounded Degree}. In: {ICALP 2019 - 46th
  International Colloquium on Automata, Languages and Programming}. pp. 1--15.
  Patras, Greece (Jul 2019), \url{https://hal.inria.fr/hal-02338498}

\bibitem{equi2020graphs}
Equi, M., M{\"a}kinen, V., Tomescu, A.I.: Graphs cannot be indexed in
  polynomial time for sub-quadratic time string matching, unless seth fails.
  In: Bure{\v{s}}, T., Dondi, R., Gamper, J., Guerrini, G., Jurdzi{\'{n}}ski,
  T., Pahl, C., Sikora, F., Wong, P.W. (eds.) SOFSEM 2021: Theory and Practice
  of Computer Science. pp. 608--622. Springer International Publishing, Cham
  (2021)

\bibitem{gibney2020simple}
Gibney, D., Hoppenworth, G., Thankachan, S.V.: Simple reductions from
  formula-sat to pattern matching on labeled graphs and subtree isomorphism.
  In: Le, H.V., King, V. (eds.) 4th Symposium on Simplicity in Algorithms,
  {SOSA} 2021, Virtual Conference, January 11-12, 2021. pp. 232--242. {SIAM}
  (2021). \doi{10.1137/1.9781611976496.26},
  \url{https://doi.org/10.1137/1.9781611976496.26}

\bibitem{NP}
Gibney, D., Thankachan, S.V.: On the hardness and inapproximability of
  recognizing wheeler graphs. In: 27th Annual European Symposium on Algorithms,
  {ESA} 2019, September 9-11, 2019, Munich/Garching, Germany. pp. 51:1--51:16
  (2019). \doi{10.4230/LIPIcs.ESA.2019.51},
  \url{https://doi.org/10.4230/LIPIcs.ESA.2019.51}

\bibitem{Hop}
Hopcroft, J.E.: An $n\log n$ algorithm for minimizing states in a finite
  automaton. Tech. rep., Stanford University (January 1971)

\bibitem{potechin2018lengths}
Potechin, A., Shallit, J.: Lengths of words accepted by nondeterministic finite
  automata. Information Processing Letters  \textbf{162},  105993 (2020).
  \doi{https://doi.org/10.1016/j.ipl.2020.105993},
  \url{https://www.sciencedirect.com/science/article/pii/S0020019020300806}

\bibitem{prezza2020locating}
Prezza, N.: On locating paths in compressed tries. In: Proceedings of the
  Thirty-Second Annual ACM-SIAM Symposium on Discrete Algorithms. p. 744–760.
  Society for Industrial and Applied Mathematics, USA (2021)

\bibitem{ordered_aut}
Shyr, H., Thierrin, G.: Ordered automata and associated languages. Tamkang J.
  Math (5),  9--20 (1974)

\bibitem{Gagie}
Travis Gagie{,} Giovanni Manzini~e Sir{\'e}n, J.: Wheeler graphs: A framework
  for bwt-based data structures. Theoretical computer science  \textbf{698},
  67--78 (2017)

\end{thebibliography}

\end{document}